\documentclass[11pt]{article}

\usepackage[margin=1in]{geometry}
\usepackage{amsmath,amssymb,amsthm,mathtools}
\usepackage{enumitem}
\usepackage{hyperref}

\hypersetup{
  colorlinks=true,
  linkcolor=blue,
  citecolor=blue,
  urlcolor=blue
}

\title{Monotonicity of Normalized Implied-Volatility Coordinates under No-Arbitrage}
\author{Jian Sun\\FTI Consulting\\jiansun@uchicago.edu}
\date{}

\numberwithin{equation}{section}

\newtheorem{theorem}{Theorem}[section]
\newtheorem{lemma}[theorem]{Lemma}
\newtheorem{proposition}[theorem]{Proposition}
\theoremstyle{remark}
\newtheorem{remark}[theorem]{Remark}

\begin{document}

\maketitle

\begin{abstract}
For a fixed maturity, an arbitrage-free option smile induces natural normalized
strike coordinates.  In the Black--Scholes setting these coordinates are closely
related to the normalizing transformations introduced by Fukasawa, equivalently
the negatives of the Black--Scholes quantities \(d_1\) and \(d_2\).  We make two
contributions.  First, we give an elementary no-arbitrage proof of monotonicity
for the central normalized coordinate \(k/v(k)\), based only on finite-strike
comparisons, convexity, monotonicity, and put--call parity.  Thus the main order
arguments are discrete: they apply directly to a finite set of quoted strikes and
do not require a continuously quoted smile, differentiability of option prices,
or differentiability of implied volatility.  Second, we prove a model-free
normal-variance identity in Bachelier implied-volatility coordinates.  Fukasawa's
lognormal identity expresses the remaining variance as an integral of implied
variance against the standard normal density after applying his normalized
coordinates.  We establish the corresponding normal identity: the expected
remaining normal variance equals the normal-density weighted integral of squared
Bachelier implied volatilities in the Bachelier normalized coordinate.  This
places the normal result alongside Fukasawa's lognormal transformation and the
variance-swap identities in the volatility-derivatives literature surveyed by
Carr and Lee.
\end{abstract}

\section{Introduction}

A fixed-maturity option smile is not an arbitrary function of strike.  Static
absence of arbitrage imposes monotonicity and convexity restrictions on call and
put prices, and these restrictions translate into nontrivial constraints on
implied volatility.  Such constraints are important both theoretically and
practically: they underlie interpolation of option surfaces, arbitrage detection,
and model-free extraction of variance-like quantities from option prices.

This paper studies one simple but useful consequence of no-arbitrage: the
monotonicity of normalized implied-volatility coordinates.  In the lognormal
case, let \(F\) be the forward price, let
\[
  k=\log(K/F), \qquad v(k)
\]
be the total Black--Scholes implied volatility, and define
\[
  f_1(k)=\frac{k}{v(k)}-\frac{v(k)}{2}, \qquad
  f_2(k)=\frac{k}{v(k)}+\frac{v(k)}{2}.
\]
These are precisely the normalizing transformations studied by
Fukasawa~\cite{Fukasawa2012}.  Since the Black--Scholes quantities are
\[
  d_1(k)=-\frac{k}{v(k)}+\frac{v(k)}{2}, \qquad
  d_2(k)=-\frac{k}{v(k)}-\frac{v(k)}{2},
\]
the monotonicity of \(f_1\) and \(f_2\) is equivalent to the monotonicity of
\(-d_1\) and \(-d_2\), or equivalently to the fact that \(d_1\) and \(d_2\)
decrease with strike.

Fukasawa's normalizing transformation paper is the closest reference to the
present note.  Fukasawa introduced nonlinear transforms of the implied volatility
smile and showed that, under no-arbitrage, they have strong monotonicity
properties.  He also used these transforms to obtain model-free skew bounds and
to express prices of European payoffs, including fair strikes of variance and
gamma swaps, in terms of the implied volatility smile.  The proof strategy in
Fukasawa is naturally formulated in a continuous-strike framework: the smile is
treated as a function on a continuum of log-strikes, and the key arguments use
regularity of the strike-price or implied-volatility functions.  This is the
right framework for deriving integral payoff formulas and skew bounds, but it is
not the way market smiles are observed.

The contribution of the present note is twofold.  The first contribution is to
isolate an elementary finite-strike order mechanism.  The main proofs compare
two strikes at a time and use only the static no-arbitrage inequalities visible
on a quoted option chain.  In particular, the arguments do not assume a
continuously quoted smile, nor do they need differentiability of call prices,
put prices, or implied volatilities.  This makes the result directly applicable
to discrete market quotes, where one first checks monotonicity and convexity of
the quoted option prices and then obtains monotonicity of the corresponding
normalized coordinates.  The same finite-strike argument also gives a direct
normal, or Bachelier, analogue.

The second contribution is a model-free normal-variance identity.  Fukasawa's
lognormal transformation leads to a striking identity: the remaining variance can
be represented as an integral of Black implied variance against the standard
normal density after changing variables to the normalized coordinate.  We prove
the corresponding identity under the Bachelier convention.  If
\(x=(F-K)/\sigma_N(K)\) denotes the Bachelier normalized coordinate, then the
expected remaining normal variance can be written as a normal-density weighted
average of \(\sigma_N(x)^2\).  This provides a normal-model counterpart to
Fukasawa's lognormal variance identity and gives an additional link between the
normalizing-coordinate viewpoint and the model-free variance-swap literature.

The broader literature is extensive.  The representation of risk-neutral
distributions through call-price convexity goes back to
Breeden and Litzenberger~\cite{BreedenLitzenberger1978}.  The Black--Scholes
and Bachelier formulas are classical references for the lognormal and normal
implied-volatility conventions~\cite{BlackScholes1973,Bachelier1900}.  The
replication and static-representation viewpoint for variance-type claims is
closely related to the log-contract and variance-swap literature of
Neuberger~\cite{Neuberger1994}, Dupire~\cite{Dupire1993},
Carr and Madan~\cite{CarrMadan1998}, and Demeterfi, Derman, Kamal and
Zou~\cite{DemeterfiDermanKamalZou1999}.  Carr and Lee~\cite{CarrLee2009}
provide a useful review of volatility derivatives, including variance swaps,
volatility swaps, VIX products, and the weighted-average representation of the
variance swap rate in implied-volatility coordinates.  Their review is also a
convenient source for many of the standard references in the volatility-
derivatives literature.  Modern arbitrage-free parametrizations of volatility
smiles include the SVI framework and its arbitrage-free variants
~\cite{Gatheral2006,GatheralJacquier2014}.  Related no-arbitrage inequalities
for Black--Scholes implied volatility appear in Tehranchi~\cite{Tehranchi2020}.
Fukasawa's transformation remains a central reference for the specific
normalized coordinates considered in this paper.

\section{No-Arbitrage Setup}

We work on a single maturity \(T\).  To avoid inessential notation, we use
forward prices and set the discount factor to one.  Thus \(F>0\) denotes the
forward price of the underlying at maturity \(T\).  Let \(C(K)\) and \(P(K)\)
denote undiscounted European call and put prices with strike \(K\).

For the lognormal model we take \(K>0\).  For the normal model it is natural to
allow \(K\in\mathbb{R}\), since the Bachelier terminal distribution has support
on the real line.

We assume the standard static no-arbitrage conditions:
\begin{enumerate}[label=(A\arabic*)]
  \item \(C(K)\) is decreasing and convex in \(K\).
  \item \(P(K)\) is increasing and convex in \(K\).
  \item Put--call parity holds:
  \[
    C(K)-P(K)=F-K .
  \]
  \item The boundary conditions are consistent with absence of static arbitrage.
\end{enumerate}
In the lognormal case the relevant boundary conditions are
\[
  C(0)=F,\qquad \lim_{K\to\infty} C(K)=0.
\]
In the normal case on the full real strike line, the corresponding limits are
\[
  \lim_{K\to-\infty}\{C(K)-(F-K)\}=0,\qquad
  \lim_{K\to\infty} C(K)=0.
\]

We write \(\phi\) and \(\Phi\) for the standard normal density and distribution
function.

\begin{remark}[Discrete market interpretation]
Although the notation below is written in functional form, the main order
results should be read discretely.  Given quoted strikes \(K_1<\cdots<K_n\),
one only needs the finite no-arbitrage inequalities on this grid: call prices
must be decreasing and convex on the grid, put prices must be increasing and
convex on the grid, and put--call parity must hold at the quoted strikes.  The
proofs compare two strikes at a time and therefore do not require interpolation
between quoted strikes, existence of a strike density, or differentiability with
respect to strike.  Continuous-strike notation is used only for convenience and
for comparison with the classical literature.
\end{remark}

\begin{lemma}[Monotonicity of normalized put price]
\label{lem:put-over-k}
In the lognormal setting, \(K\mapsto P(K)/K\) is increasing on \(K>0\).
\end{lemma}

\begin{proof}
Since \(P\) is convex and \(P(0)=0\), for \(0<K_1<K_2\) convexity gives
\[
  P(K_1)\le \frac{K_1}{K_2}P(K_2)+\left(1-\frac{K_1}{K_2}\right)P(0)
  =\frac{K_1}{K_2}P(K_2).
\]
Hence \(P(K_1)/K_1\le P(K_2)/K_2\).
\end{proof}

\section{The Normal, or Bachelier, Case}

This section proves the monotonicity result for the normal, or Bachelier,
implied-volatility coordinate.  The goal is to show that, under static
no-arbitrage, the normalized coordinate
\[
  x(K)=\frac{F-K}{\sigma_N(K)}
\]
is strictly decreasing in strike, or equivalently that
\((K-F)/\sigma_N(K)\) is strictly increasing.  The proof is deliberately
finite-strike: it compares two strikes at a time and uses only monotonicity of
quoted call and put prices together with elementary monotonicity properties of
the Bachelier pricing functions.  After proving the two-strike monotonicity
statement, we record the corresponding range property on the full real strike
line.

For a Bachelier implied volatility \(\sigma>0\), define
\[
  x=\frac{F-K}{\sigma}.
\]
The Bachelier call and put prices are
\begin{align}
  C_N(K,\sigma)
    &=\sigma\{\phi(x)+x\Phi(x)\}, \label{eq:bachelier-call}\\
  P_N(K,\sigma)
    &=\sigma\{\phi(x)-x\Phi(-x)\}. \label{eq:bachelier-put}
\end{align}

\begin{lemma}
\label{lem:bach-monotone}
The function
\[
  h(x)=\phi(x)+x\Phi(x)
\]
is strictly increasing in \(x\), and
\[
  \bar h(x)=\phi(x)-x\Phi(-x)
\]
is strictly decreasing in \(x\).
\end{lemma}

\begin{proof}
Direct differentiation gives
\[
  h'(x)=\Phi(x)>0,\qquad \bar h'(x)=-\Phi(-x)<0.
\]
\end{proof}

\begin{theorem}[Bachelier normalized coordinate]
\label{thm:bachelier-coordinate}
Let \(K_1<K_2\), and let \(\sigma_i=\sigma(K_i)\) be the Bachelier implied
volatility at \(K_i\).  Then
\[
  \frac{K_1-F}{\sigma_1}
  <
  \frac{K_2-F}{\sigma_2}.
\]
Equivalently, \(x(K)=(F-K)/\sigma(K)\) is strictly decreasing in \(K\).
\end{theorem}

\begin{proof}
It is enough to prove that \(x_1>x_2\), where
\[
  x_i=\frac{F-K_i}{\sigma_i}.
\]
We distinguish three cases.

First suppose \(K_1<F<K_2\).  Then \(x_1>0>x_2\), so the result is immediate.

Next suppose \(F<K_1<K_2\).  Then \(x_1,x_2<0\).  If \(\sigma_2<\sigma_1\),
then
\[
  \frac{K_1-F}{\sigma_1}<\frac{K_2-F}{\sigma_2},
\]
which is the desired result.  If \(\sigma_2\ge \sigma_1\), no-arbitrage gives
\(C(K_1)>C(K_2)\), while the Bachelier formula gives
\[
  C(K_i)=\sigma_i h(x_i).
\]
If \(x_1\le x_2\), then \(h(x_1)\le h(x_2)\), and hence
\[
  C(K_1)=\sigma_1h(x_1)\le \sigma_2h(x_2)=C(K_2),
\]
a contradiction.  Hence \(x_1>x_2\).

Finally suppose \(K_1<K_2<F\).  Then \(x_1,x_2>0\).  If
\(\sigma_1<\sigma_2\), the desired inequality is immediate.  If
\(\sigma_1\ge \sigma_2\), no-arbitrage gives \(P(K_1)<P(K_2)\), and
\[
  P(K_i)=\sigma_i\bar h(x_i).
\]
If \(x_1\le x_2\), then \(\bar h(x_1)\ge \bar h(x_2)\), and therefore
\[
  P(K_1)=\sigma_1\bar h(x_1)\ge \sigma_2\bar h(x_2)=P(K_2),
\]
again a contradiction.  Thus \(x_1>x_2\).
\end{proof}

\begin{remark}
The preceding proof is purely order-theoretic and finite-strike.  It uses only
monotonicity of quoted option prices in strike and monotonicity of the elementary
Bachelier pricing functions.  It does not differentiate the option-price curve
or the implied-volatility smile, and it does not require option quotes between
\(K_1\) and \(K_2\).
\end{remark}

\begin{proposition}[Range of the Bachelier coordinate]
\label{prop:range-normal}
Assume the normal smile is defined on the full real strike line and satisfies
the usual no-arbitrage boundary conditions.  Then
\[
  \lim_{K\to-\infty}\frac{F-K}{\sigma(K)}=+\infty,
  \qquad
  \lim_{K\to+\infty}\frac{F-K}{\sigma(K)}=-\infty .
\]
\end{proposition}

\begin{proof}
By Theorem~\ref{thm:bachelier-coordinate}, \(x(K)=(F-K)/\sigma(K)\) is
decreasing, so the two limits exist in the extended real line.

Suppose first that \(\lim_{K\to+\infty}x(K)=\ell>-\infty\).  Since
\(F-K=\sigma(K)x(K)\) and \(F-K\to-\infty\), this forces
\(\sigma(K)\to+\infty\).  From \eqref{eq:bachelier-call},
\[
  C(K)=\sigma(K)\{\phi(x(K))+x(K)\Phi(x(K))\}.
\]
For every finite lower bound on \(x\), the bracketed term is strictly positive.
Hence \(C(K)\) cannot converge to zero, contradicting the no-arbitrage boundary
condition.  Therefore the right limit is \(-\infty\).

The left limit follows similarly from the put formula.  If
\(\lim_{K\to-\infty}x(K)=\ell<+\infty\), then \(\sigma(K)\to+\infty\) and
\(P(K)=\sigma(K)\{\phi(x(K))-x(K)\Phi(-x(K))\}\) cannot converge to zero, a
contradiction.
\end{proof}

\section{The Lognormal, or Black--Scholes, Case}

This section proves the analogous monotonicity result in the lognormal, or
Black--Scholes, convention.  Here the normalized coordinate is
\[
  z(k)=\frac{k}{v(k)}, \qquad k=\log(K/F),
\]
where \(v(k)\) is the total Black--Scholes implied volatility.  The first goal
is to prove, again by a finite-strike no-arbitrage comparison, that \(k/v(k)\)
is increasing in log-moneyness.  We then relate this discrete monotonicity
statement to Fukasawa's continuous normalizing transformations,
\[
  f_1(k)=\frac{k}{v(k)}-\frac{v(k)}{2},
  \qquad
  f_2(k)=\frac{k}{v(k)}+\frac{v(k)}{2},
\]
which correspond to \(-d_1\) and \(-d_2\).  Thus the section separates the
finite market-quote result from the additional differentiability assumptions
needed for the continuous-strike Fukasawa formulation.

Let \(v>0\) denote total Black--Scholes implied volatility.  We use log-forward
moneyness
\[
  k=\log(K/F),
\]
and define the normalized coordinate
\[
  z=\frac{k}{v}.
\]
The Black--Scholes quantities are
\[
  d_1=-z+\frac{v}{2},\qquad d_2=-z-\frac{v}{2}.
\]
Equivalently,
\[
  -d_1=z-\frac{v}{2},\qquad -d_2=z+\frac{v}{2}.
\]
These two functions are exactly Fukasawa's normalizing transformations
\(f_1\) and \(f_2\).

The Black--Scholes call price, normalized by the forward, is
\begin{equation}
\label{eq:bs-call-normalized}
  c(z,v)=\Phi\!\left(-z+\frac{v}{2}\right)
  -e^{vz}\Phi\!\left(-z-\frac{v}{2}\right).
\end{equation}
The put price normalized by strike is
\begin{equation}
\label{eq:bs-put-normalized}
  \tilde p(z,v)=\frac{P(K)}{K}
  =\Phi\!\left(z+\frac{v}{2}\right)
  -e^{-vz}\Phi\!\left(z-\frac{v}{2}\right).
\end{equation}

\begin{lemma}
\label{lem:mills}
For all \(u\in\mathbb{R}\),
\[
  \phi(u)+u\Phi(u)>0,
  \qquad
  \phi(u)-u\Phi(-u)>0.
\]
\end{lemma}

\begin{proof}
The first expression is the normalized Bachelier call price with unit volatility
and normalized moneyness \(u\); the second is the corresponding normalized
Bachelier put price.  Both are strictly positive for finite \(u\).
\end{proof}

\begin{lemma}
\label{lem:bs-monotone}
The function \(c(z,v)\) is strictly decreasing in \(z\) and strictly increasing
in \(v\).  The function \(\tilde p(z,v)\) is strictly increasing in both \(z\)
and \(v\).
\end{lemma}

\begin{proof}
Differentiating \eqref{eq:bs-call-normalized} and using
\[
  \phi\!\left(-z+\frac{v}{2}\right)
  =
  e^{vz}\phi\!\left(-z-\frac{v}{2}\right)
\]
gives
\[
  \frac{\partial c}{\partial z}
  =
  -v e^{vz}\Phi\!\left(-z-\frac{v}{2}\right)<0
\]
and
\[
  \frac{\partial c}{\partial v}
  =
  e^{vz}\Bigg[
    \phi\!\left(-z-\frac{v}{2}\right)
    -z\Phi\!\left(-z-\frac{v}{2}\right)
  \Bigg]>0,
\]
where the final inequality follows from Lemma~\ref{lem:mills}.

Similarly,
\[
  \frac{\partial \tilde p}{\partial z}
  =
  v e^{-vz}\Phi\!\left(z-\frac{v}{2}\right)>0,
\]
and
\[
  \frac{\partial \tilde p}{\partial v}
  =
  e^{-vz}\Bigg[
    \phi\!\left(z-\frac{v}{2}\right)
    +z\Phi\!\left(z-\frac{v}{2}\right)
  \Bigg]>0.
\]
\end{proof}

\begin{theorem}[Monotonicity of \(k/v(k)\)]
\label{thm:k-over-v}
Let \(0<K_1<K_2\), and let \(v_i=v(k_i)\) be the corresponding total implied
volatilities, where \(k_i=\log(K_i/F)\).  Then
\[
  \frac{k_1}{v_1}<\frac{k_2}{v_2}.
\]
\end{theorem}

\begin{proof}
Let \(z_i=k_i/v_i\).  We prove \(z_1<z_2\).

If \(K_1<F<K_2\), then \(k_1<0<k_2\), so the claim is immediate.

Suppose \(F<K_1<K_2\), so \(k_1,k_2>0\).  If \(v_2<v_1\), then
\(z_1<z_2\) follows directly from \(k_1<k_2\).  If \(v_2\ge v_1\), no-arbitrage
gives \(C(K_1)>C(K_2)\).  Since
\[
  C(K_i)=F\,c(z_i,v_i),
\]
Lemma~\ref{lem:bs-monotone} implies that \(z_1\ge z_2\) would lead to
\[
  c(z_1,v_1)\le c(z_2,v_2),
\]
contradicting \(C(K_1)>C(K_2)\).  Hence \(z_1<z_2\).

Suppose finally that \(K_1<K_2<F\), so \(k_1,k_2<0\).  If \(v_1<v_2\), the
claim is immediate.  If \(v_1\ge v_2\), use the normalized put price.  By
Lemma~\ref{lem:put-over-k}, \(P(K)/K\) is increasing in \(K\).  Since
\[
  \frac{P(K_i)}{K_i}=\tilde p(z_i,v_i),
\]
Lemma~\ref{lem:bs-monotone} implies that \(z_1\ge z_2\) would give
\[
  \tilde p(z_1,v_1)\ge \tilde p(z_2,v_2),
\]
contradicting the monotonicity of \(P(K)/K\).  Therefore \(z_1<z_2\).
\end{proof}

\begin{remark}[Finite-strike nature of Theorem~\ref{thm:k-over-v}]
The proof of Theorem~\ref{thm:k-over-v} is the discrete counterpart of the
usual continuous-strike no-arbitrage reasoning.  For any two quoted strikes
\(K_1<K_2\), the conclusion follows from the two quoted option prices, their
put--call-parity relation, and the finite-grid monotonicity/convexity
inequalities.  No strike derivative, digital price, density, or interpolated
smile is used.  This is the main distinction from the continuous normalizing-
transformation analysis of Fukasawa.
\end{remark}

\begin{lemma}[Digital bounds]
\label{lem:digital-bounds}
Assume, in addition, that the normalized call price
\[
  c(k)=\frac{C(Fe^k)}{F}
\]
is differentiable in \(k\).  Then
\[
  -e^k\le c'(k)\le 0 .
\]
\end{lemma}

\begin{proof}
Since \(K=Fe^k\), we have \(c'(k)=e^k C_K(K)\) whenever \(C_K\) exists.
Convexity and monotonicity of the call price imply \(-1\le C_K(K)\le 0\).
The result follows.
\end{proof}

\begin{theorem}[Continuous-strike Fukasawa normalizing transformations]
\label{thm:fukasawa}
For comparison with Fukasawa's normalizing transformations, assume in this
subsection that the lognormal implied total volatility \(v(k)\) is
differentiable as a function of log-moneyness.  Under the no-arbitrage
assumptions above, the functions
\[
  f_1(k)=\frac{k}{v(k)}-\frac{v(k)}{2},
  \qquad
  f_2(k)=\frac{k}{v(k)}+\frac{v(k)}{2}
\]
are increasing in \(k\).  Equivalently, \(d_1(k)\) and \(d_2(k)\) are decreasing
functions of strike.
\end{theorem}

\begin{proof}
Write
\[
  f_1(k)=\frac{k}{v(k)}-\frac{v(k)}{2},\qquad
  f_2(k)=\frac{k}{v(k)}+\frac{v(k)}{2},
\]
so that \(f_2-f_1=v\) and \(f_2^2-f_1^2=2k\).  Hence
\[
  \phi(f_1)=e^k\phi(f_2).
\]
The normalized Black--Scholes call price is
\[
  c(k)=\Phi(-f_1(k))-e^k\Phi(-f_2(k)).
\]
Differentiating and using the identity above gives
\[
  c'(k)=e^k\{\phi(f_2(k))v'(k)-\Phi(-f_2(k))\}.
\]
By Lemma~\ref{lem:digital-bounds},
\[
  -1\le \phi(f_2(k))v'(k)-\Phi(-f_2(k))\le 0,
\]
or equivalently
\[
  -\frac{\Phi(f_2(k))}{\phi(f_2(k))}
  \le v'(k)\le
  \frac{\Phi(-f_2(k))}{\phi(f_2(k))}.
\]
Now
\[
  f_2'(k)=\frac{1-f_1(k)v'(k)}{v(k)},\qquad
  f_1'(k)=\frac{1-f_2(k)v'(k)}{v(k)}.
\]
If \(f_1(k)\le 0\), the lower bound on \(v'\) immediately gives
\(1-f_1v'>0\).  If \(f_1(k)>0\), then \(f_2(k)>f_1(k)>0\), and the upper bound
on \(v'\) gives
\[
  f_1(k)v'(k)
  \le
  f_1(k)\frac{\Phi(-f_2(k))}{\phi(f_2(k))}
  <
  f_2(k)\frac{\Phi(-f_2(k))}{\phi(f_2(k))}
  <1,
\]
where the last inequality is the usual Mills-ratio inequality
\(x\Phi(-x)<\phi(x)\) for \(x>0\).  Therefore \(f_2'(k)>0\).

The proof for \(f_1\) is analogous.  If \(f_2(k)\ge 0\), the upper bound on
\(v'\) and the inequality \(x\Phi(-x)<\phi(x)\) imply \(f_2(k)v'(k)<1\).  If
\(f_2(k)<0\), the lower bound on \(v'\) gives
\[
  f_2(k)v'(k)
  \le
  -f_2(k)\frac{\Phi(f_2(k))}{\phi(f_2(k))}
  <1,
\]
which is again the Mills-ratio inequality applied to \(x=-f_2(k)>0\).  Hence
\(f_1'(k)>0\).
\end{proof}

\begin{remark}[Relation to Fukasawa and to the discrete proof]
Fukasawa~\cite{Fukasawa2012} is the foundational reference for the
normalizing transformations \(f_1=-d_1\) and \(f_2=-d_2\).  His analysis treats
the implied volatility smile as a continuous object and uses regularity of the
strike-price or implied-volatility functions in order to obtain skew bounds and
payoff-representation formulas.  The differentiable proof above is included
only to connect the notation of this paper to that continuous-strike framework.

The independent contribution of the present note is different.  Theorems
\ref{thm:bachelier-coordinate} and \ref{thm:k-over-v} are proved directly from
finite-strike comparisons.  They apply to the market situation in which one
observes only a discrete list of strikes and option prices.  In that setting,
there may be no meaningful derivative of the quoted smile, and even continuity
between strikes is an interpolation choice rather than market data.  The
monotonicity of the normalized coordinate is therefore a consequence of discrete
static no-arbitrage itself, not of smoothness of a continuously specified smile.
\end{remark}

\section{A Model-Free Normal-Variance Identity}

This section proves the second main contribution of the paper.  Fukasawa's
normalizing transformation gives a lognormal model-free variance identity: after
passing to the normalized Black coordinate, remaining variance is represented as
a standard-normal weighted integral of implied variance.  We now prove the
Bachelier analogue.  The result is model-free in the same sense: it depends only
on the option-implied terminal law and not on a particular diffusion model for
volatility.

We now return to the normal setting.  Let the underlying be a martingale
satisfying
\[
  dS_t=\sigma_t\,dW_t,
\]
where \(\sigma_t\) may be stochastic.  Then
\[
  d(S_t-F)^2=2(S_t-F)\sigma_t\,dW_t+\sigma_t^2\,dt,
\]
and hence
\begin{equation}
\label{eq:normal-integrated-variance}
  \mathbb{E}\!\left[(S_T-F)^2\right]
  =
  \mathbb{E}\!\left[\int_0^T\sigma_t^2\,dt\right].
\end{equation}
Thus the expected integrated normal variance is model-free once the
risk-neutral distribution of \(S_T\) is known.

Assume for this section that the call price is twice differentiable in strike
and that boundary terms vanish.  Then the Breeden--Litzenberger identity gives
\[
  C_{KK}(K)=q(K),
\]
where \(q\) is the risk-neutral density of \(S_T\).  Therefore
\begin{equation}
\label{eq:second-moment-density}
  \mathbb{E}\!\left[(S_T-F)^2\right]
  =
  \int_{\mathbb{R}}(K-F)^2 C_{KK}(K)\,dK.
\end{equation}

Let \(\sigma_N(K)\) be the Bachelier implied volatility and set
\[
  x(K)=\frac{F-K}{\sigma_N(K)}.
\]
By Theorem~\ref{thm:bachelier-coordinate}, \(x(K)\) is strictly decreasing.
Assume it maps the real strike line onto \(\mathbb{R}\), as in
Proposition~\ref{prop:range-normal}.  We may therefore write
\(\sigma_N(x)\) for the implied volatility expressed in the normalized
coordinate.

\begin{theorem}[Normal variance in Bachelier-implied coordinates]
\label{thm:normal-variance}
Under the regularity and boundary assumptions stated above,
\[
  \mathbb{E}\!\left[(S_T-F)^2\right]
  =
  \int_{\mathbb{R}}\sigma_N(x)^2\phi(x)\,dx.
\]
\end{theorem}

\begin{proof}
The Bachelier call price can be written as
\[
  C(K)=\sigma_N(K)\{\phi(x(K))+x(K)\Phi(x(K))\}.
\]
Since \(F-K=\sigma_N(K)x(K)\), differentiation gives
\[
  \sigma_N'(K)x(K)+\sigma_N(K)x'(K)=-1.
\]
A direct calculation then yields
\[
  C_K(K)=\sigma_N'(K)\phi(x(K))-\Phi(x(K)).
\]
Using integration by parts in \eqref{eq:second-moment-density},
\[
  \int_{\mathbb{R}}(K-F)^2C_{KK}(K)\,dK
  =
  -2\int_{\mathbb{R}}(K-F)C_K(K)\,dK.
\]
Substituting the expression for \(C_K\), using \(K-F=-\sigma_N x\), and then
using \(\sigma_N'x+\sigma_N x'=-1\), the integrand simplifies to
\[
  -\sigma_N(K)^2\phi(x(K))x'(K).
\]
Therefore
\[
  \mathbb{E}\!\left[(S_T-F)^2\right]
  =
  -\int_{\mathbb{R}}\sigma_N(K)^2\phi(x(K))x'(K)\,dK.
\]
Finally, since \(x(K)\) is decreasing from \(+\infty\) to \(-\infty\), the
change of variables \(x=x(K)\) gives
\[
  -\int_{\mathbb{R}}\sigma_N(K)^2\phi(x(K))x'(K)\,dK
  =
  \int_{\mathbb{R}}\sigma_N(x)^2\phi(x)\,dx.
\]
\end{proof}

\begin{remark}[Comparison with Fukasawa's lognormal identity]
Theorem~\ref{thm:normal-variance} is the normal-model analogue of Fukasawa's
lognormal variance identity.  In Fukasawa's setting, the normalized Black
coordinates transform the remaining variance into an integral of Black implied
variance against the standard normal density.  Here the Bachelier coordinate
\(x=(F-K)/\sigma_N(K)\) plays the same role, and the remaining normal variance
is represented as
\[
  \int_{\mathbb{R}} \sigma_N(x)^2\phi(x)\,dx .
\]

\end{remark}

\section{Conclusion}

The main message is that normalized implied-volatility coordinates have both a
discrete no-arbitrage origin and a model-free variance interpretation.  Unlike
continuous-strike proofs based on digital bounds or derivatives of the
implied-volatility smile, the central monotonicity argument here is discrete: it
compares two quoted strikes and uses only finite-grid monotonicity, convexity,
and put--call parity.  In the Black--Scholes case, this gives monotonicity of
the normalized coordinate \(k/v(k)\), which is the midpoint of Fukasawa's two
normalizing transformations \(f_1=-d_1\) and \(f_2=-d_2\).  In the Bachelier
case, the same argument yields monotonicity of \((F-K)/\sigma_N(K)\).  The
paper's second main result shows that this Bachelier coordinate also gives a
model-free normal-variance identity, paralleling Fukasawa's lognormal variance
identity and expressing expected remaining normal variance as a standard-normal
weighted integral of squared Bachelier implied volatility.

The results are useful for three reasons.  First, they provide simple
arbitrage-diagnostic checks for implied-volatility interpolation.  Second, they
clarify why the \(d_1,d_2\) coordinates are natural for organizing a smile.
Third, they show that the normal and lognormal conventions share the same
underlying no-arbitrage mechanism.

\end{document}